\documentclass[a4paper,UKenglish,autoref,thm-restate]{lipics-v2019}
\nolinenumbers
\usepackage[ruled,noend,vlined,linesnumbered]{algorithm2e} 
\usepackage{complexity}

\usepackage{tikz}
\usetikzlibrary{arrows,calc}
\tikzstyle{vertex} = [circle,fill=black!0,minimum size=4pt,inner sep=1pt]
\tikzstyle{every path} = [bend angle=30,>=stealth]

\newcommand{\breakpoints}[1]{\ensuremath{b(#1)}}

\newcommand{\SBPBIshort}{{\sc sbpbi}}

\newcommand{\decisionproblem}[3]{
\begin{center} 
\noindent\fbox{\parbox{0.97\textwidth}{
\begin{minipage}[t]{1\linewidth}
\textsc{#1}

Input:\ \ \ \ \ \ \  {#2}

Question: {#3}
\end{minipage}
}
}    
\end{center}
}

\title{Sorting by Prefix Block-Interchanges}
\author{
Anthony Labarre
}{
Laboratoire d'Informatique Gaspard Monge, Universit\'e Gustave Eiffel, LIGM (UMR 8049), CNRS, ENPC, ESIEE Paris, UPEM, F-77454, Marne-la-Vall\'ee, France \and \url{http://igm.univ-mlv.fr/~alabarre/}
}{
Anthony.Labarre@u-pem.fr
}{
https://orcid.org/0000-0002-9945-6774
}{}
\authorrunning{
A. Labarre
}
\Copyright{
Anthony Labarre
} 

\begin{document}

\maketitle

\begin{abstract}
We initiate the study of sorting permutations using \emph{prefix block-interchanges}, which exchange any prefix of a permutation with another non-intersecting interval. The goal is to transform a given permutation into the identity permutation using as few such operations as possible. We give a 2-approximation algorithm for this problem, show how to obtain improved lower and upper bounds on the corresponding distance, 
and determine the largest possible value for that distance.
\end{abstract}

\section{Introduction}

The problem of transforming two sequences into one another using a specified set of operations has received a lot of attention in the last decades, with applications in computational biology as \emph{(genome) rearrangement problems}~\cite{fertin-combinatorics} as well as interconnection network design~\cite{lakshmivarahan-symmetry}. In the context of permutations, it can be equivalently formulated as follows:  given a permutation $\pi$ of $[n]=\{1, 2, \ldots, n\}$ and a generating set $S$ (also consisting of permutations of $[n]$), find a minimum-length sequence of elements from $S$ that sorts $\pi$. The problem is known to be \NP-hard in general~\cite{jerrum-complexity} and \W[1]-hard when parameterised by the length of a solution~\cite{Cai1997}, but some families of operations that are important in applications lead to problems that can be solved in polynomial time (e.g. \emph{exchanges}~\cite{knuth-art}, \emph{block-interchanges}~\cite{Christie1996} and \emph{signed reversals}~\cite{DBLP:journals/jacm/HannenhalliP99}), while other families yield hard problems that admit good approximations (e.g. 11/8 for \emph{reversals}~\cite{DBLP:conf/esa/BermanHK02} and for \emph{block-transpositions}~\cite{Elias2006b}).

Several restrictions of these families have also been studied, one of which stands out in the field of interconnection network design: the so-called \emph{prefix constraint}, which forces operations to act on a prefix of the permutation rather than on an arbitrary interval. Those restrictions were introduced as a way of reducing the size of the generated network while maintaining a low value for its \emph{diameter}, thereby guaranteeing a low maximum communication delay~\cite{lakshmivarahan-symmetry}.  The most famous example is perhaps the restriction of reversals (which reverse the order of elements along an interval) to \emph{prefix reversals}, and the corresponding problem known as \emph{pancake flipping}, introduced in \cite{10.2307/2318260} and whose complexity was only settled  thirty years later~\cite{DBLP:journals/jcss/BulteauFR15}. 

As \autoref{tab:results-summary} shows (see \cite{fertin-combinatorics} for undefined terms), although sorting problems using interval transformations are now fairly well understood, progress on the corresponding prefix sorting problems has been lacking, with only two families whose status has been settled and no approximation ratio smaller than $2$ for those problems not known to be in \P. As a result, while the topology of the \emph{Cayley graph} generated by those operations might present attractive properties, efficient routing algorithms (which achieve exactly the same task as the sorting algorithms in genome rearrangements) are still needed for the network to be of practical interest.

\begin{table}[htbp]
\caption{Complexity of some sorting problems on permutations in the unrestricted setting and under the prefix constraint.}
\centering
\begin{tabular}{lll}
Operation & Unrestricted & Prefix-constrained \\
\hline
reversal        & \NP-hard~\cite{caprara-sorting} & \NP-hard~\cite{DBLP:journals/jcss/BulteauFR15} \\
signed reversal & in \P~\cite{DBLP:journals/jacm/HannenhalliP99} & open \\ 
double cut-and-join & \NP-hard~\cite{DBLP:journals/jco/Chen13} & open \\
signed double cut-and-join & in \P~\cite{DBLP:journals/bioinformatics/YancopoulosAF05} & open \\
exchange & in \P~\cite{knuth-art} & in \P~\cite{akers-star} \\
block-transposition & \NP-hard~\cite{DBLP:journals/siamdm/BulteauFR12} & open \\
block-interchange & in \P~\cite{Christie1996} & open \\
\hline
\end{tabular}
\label{tab:results-summary}
\end{table}

In this work, we choose to focus on the family of block-interchanges for the following reasons:
\begin{enumerate}
 \item Along with double cut-and-joins, they constitute one of the most general kind of operations on permutations, including both exchanges and block-transpositions as special cases;
 \item Their behaviour in the unrestricted setting is understood well enough that we can hope for the corresponding prefix sorting problem to be in \P;
 \item Knowledge about these operations in the prefix setting is lacking and will be needed for more general studies; for instance, rearrangement problems on strings are usually \NP-hard, and efficient algorithms to solve them exactly or approximately routinely rely on techniques developed for permutations~\cite[part II]{fertin-combinatorics}, which currently do not exist for prefix block-interchanges.
 \end{enumerate}
To the best of our knowledge, the only published work on prefix block-interchanges is by \cite{DBLP:conf/intcompsymp/ChouYCL14}, who studied them on strings and showed that binary strings can be sorted in linear time, whereas transforming two binary strings into one another using the minimum number of prefix block-interchanges is \NP-complete. Our contributions are as follows: we prove tight upper and lower bounds on the so-called \emph{prefix block-interchange distance}; we give an approximation algorithm which we prove to be a 2-approximation with respect to two different measures; we  
show how to tighten those bounds;  
and finally, we prove that the maximum value of the distance, an important parameter in some applications~\cite{lakshmivarahan-symmetry}, is $\lfloor 2n/3\rfloor$.

\section{Notation and definitions}

A \emph{permutation} is a bijective application of a set (usually $[n]=\{1, 2, \ldots, n\}$ in this work) onto itself. The \emph{symmetric group} $S_n$ is the set of all permutations of $[n]$ together with the usual function composition applied from right to left. We write permutations using lower case Greek letters, viewing them as sequences $\pi=\langle\pi_1\ \pi_2\ \cdots\ \pi_n\rangle$, where  $\pi_i=\pi(i)$, and occasionally rely on the two-line notation to denote them. The  permutation $\iota=\langle 1\ 2\ \cdots\ n\rangle$ is the \emph{identity permutation}.

Permutations are well-known to decompose in a single way into \emph{disjoint cycles} (up to the ordering of cycles and of elements within each cycle), leading to another notation for $\pi$ based on its \emph{disjoint cycle decomposition}. For instance, when $\pi=\langle 7\ 1\ 4\ 5\ 3\ 2\ 6\rangle$, the disjoint cycle notation is $\pi=(1,7,6,2)(3,4,5)$. The \emph{conjugate} of a permutation $\pi$ by a permutation $\sigma$, both in $S_n$, is the permutation $\pi^{\sigma}=\sigma\pi\sigma^{-1}$. All permutations in $S_n$ that can be obtained from one another using this operation 
form a \emph{conjugacy class} (of $S_n$), and have the same cycle structure.

\begin{definition}\label{def:block-interchange}
\cite{Christie1996} 
The \emph{block-interchange} $\beta(i, j, k, \ell)$ with $1\leq i<j\leq k<\ell\leq n+1$ 
is the permutation that exchanges the closed intervals determined respectively by $i$ and $j-1$ and by $k$ and $\ell-1$:
\[
\left\langle
\begin{array}{l}
1\ \cdots\ i-1\ \fbox{$i\ \cdots\ j-1$}\ j\ j+1\ \cdots\ k-1\ \fbox{$k\ \cdots\ \ell-1$}\ \ell\ \ell+1\ \cdots\ n \\
\raisebox{-.05in}{$1\ \cdots\ i-1\ \fbox{$k\ \cdots\ \ell-1$}\ j\ j+1\ \cdots\ k-1\ \fbox{$i\ \cdots\ j-1$}\ \ell\ \ell+1\ \cdots\ n$} \\
\end{array}
\right\rangle.
\] 
\end{definition}

Block-interchanges generalise several well-studied operations: when $j=k$, the resulting operation exchanges two adjacent intervals, and is known as a \emph{(block-)transposition}~\cite{bafna-transpositions}; when $j=i+1$ and $\ell=k+1$, the resulting operation swaps elements in respective positions $i$ and $k$, and is called an \emph{exchange} (or \emph{(algebraic) transposition}); finally, when $i=1$, the resulting operation is called a \emph{prefix block-interchange}; prefix block-transpositions and prefix exchanges are defined analogously. We study the following problem.

\decisionproblem{sorting by prefix block-interchanges (sbpbi)}{a permutation $\pi$ in $S_n$, a number $K\in\mathbb{N}$.}{
is there a sequence of at most $K$ prefix block-interchanges that sorts $\pi$?
}

The length of a shortest sorting sequence of prefix block-interchanges for a permutation $\pi$ is its \emph{(prefix block-interchange) distance}, which we denote $pbid(\pi)$. Distances based on other operations are defined similarly.

\section{A 2-approximation based on the breakpoint graph}\label{sec:first-2-approx}

We give in this section a 2-approximation algorithm for {\SBPBIshort} based on the \emph{breakpoint graph}. We first use this structure in \autoref{sec:ub} to derive an upper bound on $pbid$ and present our algorithm, then derive a lower bound in \autoref{sec:lb} which allows us to prove its performance guarantee. The breakpoint graph is well-known to be equivalent~\cite{labarre-new} to another structure known as the \emph{cycle graph}~\cite{bafna-transpositions}, 
which allows us to use results based on either graph 
indifferently.

\begin{definition}\label{def:breakpoint-graph}
\cite{DBLP:journals/jacm/HannenhalliP99} 
For any $\pi$ in $S_n$, let $\pi'$ be the permutation of $\{0, 1, 2, \ldots, 2n+1\}$ defined by $\pi'_0=0$, $\pi'_{2n+1}=2n+1$, and $(\pi'_{2i-1}, \pi'_{2i})=(2\pi_{i}-1, 2\pi_{i})$ for $1\le i\le n$. 
The \emph{breakpoint graph} of $\pi$ is the undirected edge-bicoloured graph $G(\pi)=(V, E_b\cup E_g)$ whose vertex set is formed by the elements of $\pi'$ ordered by position and whose edge set consists of:
\begin{itemize}
 \item $E_b=\{\{\pi'_{2i}, \pi'_{2i+1}\}\ |\ \forall\ 0\le i\le n\}$, called the set of \emph{black} edges;
 \item $E_g=\{\{{2i}, {2i+1}\}\ |\ \forall\ 0\le i\le n\}$, called the set of \emph{grey} edges.
\end{itemize} 
\end{definition}

\autoref{fig:breakpoint-graph} shows an example of a breakpoint graph. Since $G(\pi)$ is $2$-regular, it decomposes in a single way into edge-disjoint cycles
which alternate black and grey edges. 
The \emph{length} of 
a cycle in $G(\pi)$ is the number of black edges it contains, and a \emph{$k$-cycle} in $G(\pi)$ is 
a 
cycle of length $k$. We let $c(G(\pi))$ (resp. $c_k(G(\pi))$) denote the number of cycles (resp. $k$-cycles) in $G(\pi)$, and refer to cycles of length one as \emph{trivial cycles}.

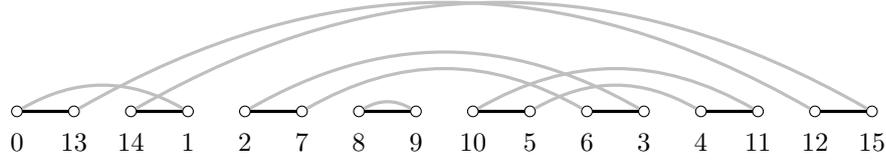
\begin{figure}[htbp]
\centering
\begin{tikzpicture}[scale=.75]
    \foreach [count=\i] \color in {black,black,black,black,black,black,black,black}
        \draw[color=\color,very thick] (2*\i-2, 0) -- (2*\i-1, 0);
    \foreach [count=\i from 0] \name in {0,13,14,1,2,7,8,9,10,5,6,3,4,11,12,15}
    \node[vertex] (\name) at (\i,0) [draw,circle] 
    [label=below:$\strut\name$]
    {};

    \begin{scope}[color=gray!50, very thick]
        \draw (0) to [bend left] (1);
        \draw (2) to [bend left] (3);
        \draw (4) to [bend right] (5);
        \draw (6) to [bend right] (7);
         \draw (8) to [bend left] (9);
        \draw (10) to [bend left] (11);
        \draw (12) to [bend right] (13);
        \draw (14) to [bend left] (15);
    \end{scope}

\end{tikzpicture}
\caption{The breakpoint graph of $\langle$7 1 4 5 3 2 6$\rangle$.}
\label{fig:breakpoint-graph}
\end{figure}

A crucial insight of strategies based on the breakpoint graph is the observation that the transformations that we apply never affect grey edges, whereas they ``cut'' black edges and replace them with new black edges. This point of view conveniently allows us to define block-interchanges in terms of the black edges on which they \emph{act}: using the notation $b_i=\{\pi_{2i-2}, \pi_{2i-1}\}$ for a black edge, a quadruplet $(b_i, b_j, b_k, b_\ell)$ of black edges with $i<j\le k<\ell$ naturally defines the block-interchange $\beta(i, j, k, \ell)$ and conversely. 

\subsection{An upper bound based on the breakpoint graph}\label{sec:ub}

The following quantity, defined for any $\pi$ in $S_n$, has been shown to be a tight\footnote{Here and in the rest of the text, ``tight'' means that equality is achieved by \emph{some} but \emph{not all} instances.} lower bound on the prefix block-transposition distance~\cite{labarre-lower} :
\begin{equation}\label{eqn:labarre-lower-bound-ptd}
g(\pi)= \frac{n+1+c(G(\pi))}{2} - c_1(G(\pi))-\left\{
\begin{array}{ll}
0 & \mbox{if } \pi_1=1,\\
1 & \mbox{otherwise.}
\end{array}
\right. 
\end{equation}

We prove in \autoref{thm:first-upper-bound-on-pbid} that this quantity is also an \emph{upper} bound on the prefix block-interchange distance. To that end, we use the following notation, based on the one introduced in \cite{bafna-transpositions}; for any two permutations $\pi$ and $\sigma$, define:
\begin{itemize}
 \item $\Delta c(\pi,\sigma)=c(G(\sigma))-c(G(\pi))$,
 \item $\Delta c_1(\pi,\sigma)=c_1(G(\sigma))-c_1(G(\pi))$,
 \item $\Delta f(\pi,\sigma)=f(\sigma)-f(\pi)$, where $f(\pi)=0$ if $\pi$ \emph{fixes} 1 (i.e. $\pi_1=1$) and $1$ otherwise, and 
 \item $\Delta g(\pi,\sigma)=g(\sigma)-g(\pi)$.
\end{itemize}
These parameters allow us to obtain the following expression, which will be useful in our proofs:
\begin{equation}\label{eqn:delta-g}
\Delta g(\pi,\sigma)=\Delta c(\pi,\sigma)/2-\Delta c_1(\pi,\sigma)-\Delta f(\pi,\sigma).
\end{equation}

We start by proving in \autoref{lemma:pi-neq-1} the existence of a prefix block-interchange that decreases $g(\pi)$ by at least one if $\pi_1\neq 1$. The proof uses the following structural result, where grey edges $\{\pi'_a, \pi'_b\}$ and $\{\pi'_c, \pi'_d\}$ (with $a<b$ and $c<d$) are said to \emph{intersect} if $a<c<b<d$ or $c<a<d<b$.

\begin{lemma}\label{every-grey-edge-intersects-another-one}
 \cite{DBLP:journals/jacm/HannenhalliP99} For every permutation $\pi$, let $e$ be a grey edge in a nontrivial cycle of $G(\pi)$; then there exists another grey edge $e'$ in $G(\pi)$ that intersects $e$.
\end{lemma}

We refer to the grey edge of $G(\pi)$ that contains $\pi'_1$ as the \emph{first} grey edge, and to the cycle that contains $0$ as the \emph{leftmost cycle}. Our figures represent alternating subpaths (i.e., paths that alternate black and grey edges) as dotted edges; therefore, such a dotted edge might correspond to a single grey edge, or to a black edge framed by two grey edges, and so on. 

\begin{lemma}\label{lemma:pi-neq-1}
 For any $\pi$ in $S_n$: if $\pi_1\neq 1$, then there exists a prefix block-interchange $\beta$ such that $\Delta c(\pi, \pi\beta)= 2$, $\Delta c_1(\pi, \pi\beta)\ge 2$, and $\Delta g(\pi, \pi\beta)\le -1$.
\end{lemma}
\begin{proof}
\autoref{every-grey-edge-intersects-another-one} guarantees the existence of a grey edge $e'$ that intersects the first grey edge; moreover, the endpoints of $e'$ ordered by position connect elements whose values are either in decreasing (case 1 below) or increasing (case 2 below) order. In both cases, if $e'$ belongs to the leftmost cycle, then there exists a prefix block-interchange that extracts two $1$-cycles (we distinguish an additional third case where $e'$ and the first grey edge share the endpoints of a black edge):
 
 \begin{enumerate}
  \item 
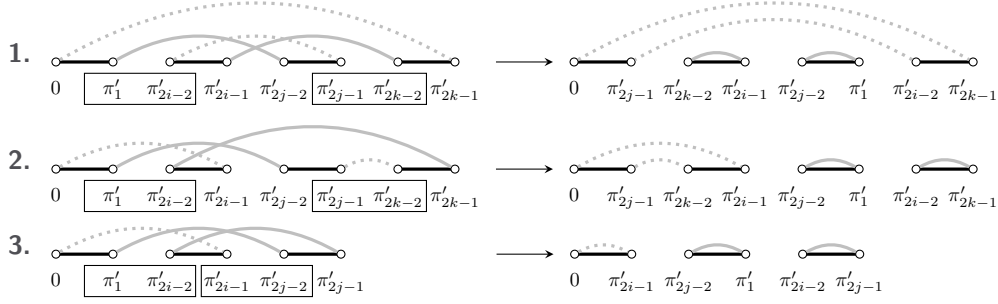
\begin{tikzpicture}[scale=.75,transform shape,baseline]

    \foreach [count=\i] \color in {black,black,black,black}
        \draw[color=\color,very thick] (2*\i-2, 0) -- (2*\i-1, 0);

    \foreach [count=\i from 0] \name in {a,b,c,d,e,f,g,h}
    \node[vertex] (\name) at (\i,0) [draw,circle] 
    {};
    \foreach \name/\lab in {a/$0$,b/$\pi'_1$,c/$\pi'_{2i-2}$,d/$\pi'_{2i-1}$,e/$\pi'_{2j-2}$,f/$\pi'_{2j-1}$,g/$\pi'_{2k-2}$,h/$\pi'_{2k-1}$}
    \node (\name2) at (\name) [label=below:\strut\lab] {};


    \begin{scope}[color=gray!50, very thick]
        \draw[] (d) to [bend left] (g);
        \draw[dotted] (a) to [bend left] (h);
        \draw[] (b) to [bend left] (e);
        \draw[dotted] (c) to [bend left] (f);
    \end{scope}
    
    \draw (.5, -.75) rectangle (2.45, -.2);
    \draw (4.5, -.75) rectangle (6.5, -.2);
\end{tikzpicture} 
\begin{tikzpicture}[scale=.75,transform shape,baseline]
\draw[->, >=stealth] (0, 0) -- (1, 0);
\end{tikzpicture}
\begin{tikzpicture}[scale=.75,transform shape,baseline]

    \foreach [count=\i] \color in {black,black,black,black}
        \draw[color=\color,very thick] (2*\i-2, 0) -- (2*\i-1, 0);

    \foreach [count=\i from 0] \name in {a,f,g,d,e,b,c,h}
    \node[vertex] (\name) at (\i,0) [draw,circle] 
    {};
    \foreach \name/\lab in {a/$0$,b/$\pi'_1$,c/$\pi'_{2i-2}$,d/$\pi'_{2i-1}$,e/$\pi'_{2j-2}$,f/$\pi'_{2j-1}$,g/$\pi'_{2k-2}$,h/$\pi'_{2k-1}$}
    \node at (\name) [label=below:\strut\lab] {};


    \begin{scope}[color=gray!50, very thick]
        \draw[] (d) to [bend right] (g);
        \draw[dotted] (a) to [bend left] (h);
        \draw[] (b) to [bend right] (e);
        \draw[dotted] (c) to [bend right] (f);
        
    \end{scope}
\end{tikzpicture} 

\item 

\begin{tikzpicture}[scale=.75,transform shape,baseline]

    \foreach [count=\i] \color in {black,black,black,black}
        \draw[color=\color,very thick] (2*\i-2, 0) -- (2*\i-1, 0);

    \foreach [count=\i from 0] \name in {a,b,c,d,e,f,g,h}
    \node[vertex] (\name) at (\i,0) [draw,circle] 
    {};
    \foreach \name/\lab in {a/$0$,b/$\pi'_1$,c/$\pi'_{2i-2}$,d/$\pi'_{2i-1}$,e/$\pi'_{2j-2}$,f/$\pi'_{2j-1}$,g/$\pi'_{2k-2}$,h/$\pi'_{2k-1}$}
    \node (\name2) at (\name) [label=below:\strut\lab] {};


    \begin{scope}[color=gray!50, very thick]
        \draw[dotted] (f) to [bend left] (g);
        \draw[dotted] (a) to [bend left] (d);
        \draw[] (b) to [bend left] (e);
        \draw[] (c) to [bend left] (h);
    \end{scope}
    
    \draw (.5, -.75) rectangle (2.45, -.2);
    \draw (4.5, -.75) rectangle (6.5, -.2);
\end{tikzpicture} 
\begin{tikzpicture}[scale=.75,transform shape,baseline]
\draw[->, >=stealth] (0, 0) -- (1, 0);
\end{tikzpicture} 
\begin{tikzpicture}[scale=.75,transform shape,baseline]

    \foreach [count=\i] \color in {black,black,black,black}
        \draw[color=\color,very thick] (2*\i-2, 0) -- (2*\i-1, 0);

    \foreach [count=\i from 0] \name in {a,f,g,d,e,b,c,h}
    \node[vertex] (\name) at (\i,0) [draw,circle] 
    {};
    \foreach \name/\lab in {a/$0$,b/$\pi'_1$,c/$\pi'_{2i-2}$,d/$\pi'_{2i-1}$,e/$\pi'_{2j-2}$,f/$\pi'_{2j-1}$,g/$\pi'_{2k-2}$,h/$\pi'_{2k-1}$}
    \node at (\name) [label=below:\strut\lab] {};


    \begin{scope}[color=gray!50, very thick]
        \draw[dotted] (f) to [bend left] (g);
        \draw[dotted] (a) to [bend left] (d);
        \draw[] (b) to [bend right] (e);
        \draw[] (c) to [bend left] (h);
    \end{scope}
\end{tikzpicture} 
\item \begin{tikzpicture}[scale=.75,transform shape,baseline]

    \foreach [count=\i] \color in {black,black,black}
        \draw[color=\color,very thick] (2*\i-2, 0) -- (2*\i-1, 0);

    \foreach [count=\i from 0] \name in {a,b,c,d,e,f}
    \node[vertex] (\name) at (\i,0) [draw,circle] 
    {};
    \foreach \name/\lab in {a/$0$,b/$\pi'_1$,c/$\pi'_{2i-2}$,d/$\pi'_{2i-1}$,e/$\pi'_{2j-2}$,f/$\pi'_{2j-1}$}
    \node (\name2) at (\name) [label=below:\strut\lab] {};


    \begin{scope}[color=gray!50, very thick]
        \draw[dotted] (a) to [bend left] (d);
        \draw[] (b) to [bend left] (e);
        \draw[] (c) to [bend left] (f);
    \end{scope}
    
    \draw (.5, -.75) rectangle (2.45, -.2);
    \draw (2.55, -.75) rectangle (4.5, -.2);
\end{tikzpicture} 
\hspace{1.4cm}
\begin{tikzpicture}[scale=.75,transform shape,baseline]
\draw[->, >=stealth] (0, 0) -- (1, 0);
\end{tikzpicture} 
\begin{tikzpicture}[scale=.75,transform shape,baseline]

    \foreach [count=\i] \color in {black,black,black}
        \draw[color=\color,very thick] (2*\i-2, 0) -- (2*\i-1, 0);

    \foreach [count=\i from 0] \name in {a,d,e,b,c,f}
    \node[vertex] (\name) at (\i,0) [draw,circle] 
    {};
    \foreach \name/\lab in {a/$0$,b/$\pi'_1$,c/$\pi'_{2i-2}$,d/$\pi'_{2i-1}$,e/$\pi'_{2j-2}$,f/$\pi'_{2j-1}$}
    \node at (\name) [label=below:\strut\lab] {};


    \begin{scope}[color=gray!50, very thick]
        \draw[dotted] (a) to [bend left] (d);
        \draw[] (b) to [bend right] (e);
        \draw[] (c) to [bend left] (f);
    \end{scope}
\end{tikzpicture} 
 \end{enumerate}

Otherwise, all edges that intersect the first grey edge belong to cycles other than the leftmost cycle. One of those edges, which belongs to some cycle $C$, must be in the same configuration as in case 1 above: indeed, if $C$ contains a grey edge that takes us from the interval covered by the first grey edge $(1, y)$ to a vertex located after $y$, then $C$ must also contain a grey edge that takes us back before $y$. 
Without loss of generality, we thus assume that $e'$ is such an edge, and therefore a prefix block-interchange that extracts two new $1$-cycles can also be applied:
 
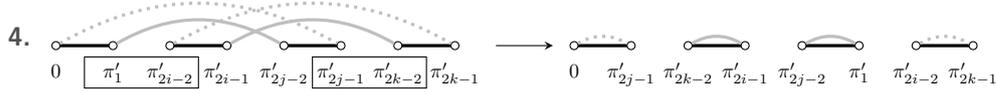
\begin{enumerate}\setcounter{enumi}{3} 
  \item

\begin{tikzpicture}[scale=.75,transform shape,baseline]

    \foreach [count=\i] \color in {black,black,black,black}
        \draw[color=\color,very thick] (2*\i-2, 0) -- (2*\i-1, 0);

    \foreach [count=\i from 0] \name in {a,b,c,d,e,f,g,h}
    \node[vertex] (\name) at (\i,0) [draw,circle] 
    {};
    \foreach \name/\lab in {a/$0$,b/$\pi'_1$,c/$\pi'_{2i-2}$,d/$\pi'_{2i-1}$,e/$\pi'_{2j-2}$,f/$\pi'_{2j-1}$,g/$\pi'_{2k-2}$,h/$\pi'_{2k-1}$}
    \node at (\name) [label=below:\strut\lab] {};


    \begin{scope}[color=gray!50, very thick]
        \draw[] (d) to [bend left] (g);
        \draw[dotted] (a) to [bend left] (f);
        \draw[] (b) to [bend left] (e);
        \draw[dotted] (c) to [bend left] (h);
        
    \end{scope}
    \draw (.5, -.75) rectangle (2.5, -.2);
    \draw (4.5, -.75) rectangle (6.5, -.2);
\end{tikzpicture} 
\begin{tikzpicture}[scale=.75,transform shape,baseline]
\draw[->, >=stealth] (0, 0) -- (1, 0);
\end{tikzpicture}
\begin{tikzpicture}[scale=.75,transform shape,baseline]

    \foreach [count=\i] \color in {black,black,black,black}
        \draw[color=\color,very thick] (2*\i-2, 0) -- (2*\i-1, 0);

    \foreach [count=\i from 0] \name in {a,f,g,d,e,b,c,h}
    \node[vertex] (\name) at (\i,0) [draw,circle] 
    {};
    \foreach \name/\lab in {a/$0$,b/$\pi'_1$,c/$\pi'_{2i-2}$,d/$\pi'_{2i-1}$,e/$\pi'_{2j-2}$,f/$\pi'_{2j-1}$,g/$\pi'_{2k-2}$,h/$\pi'_{2k-1}$}
    \node at (\name) [label=below:\strut\lab] {};


    \begin{scope}[color=gray!50, very thick]
        \draw[] (d) to [bend right] (g);
        \draw[dotted] (a) to [bend left] (f);
        \draw[] (b) to [bend right] (e);
        \draw[dotted] (c) to [bend left] (h);
        
    \end{scope}
\end{tikzpicture} 
 \end{enumerate}

In all four cases, we have $\Delta c(\pi, \pi\beta)=2$ and $\Delta c_1(\pi,  \pi\beta)\ge 2$. The exact value of $\Delta g(\pi, \pi\beta)$ will vary depending on the above configurations and is computed using  \autoref{eqn:delta-g}:
 \begin{enumerate}
  \item the leftmost cycle in $G(\pi\beta)$ is nontrivial, so $\Delta c_1(\pi, \pi\beta)=2$, $\Delta f(\pi, \pi\beta)=0$ and $\Delta g(\pi, \pi\beta)=-1$; 
\item for exactly the same reasons as case 1, we obtain $\Delta g(\pi, \pi\beta)=-1$; 
\item we have two possible subcases: 
\begin{enumerate}
 \item if the leftmost cycle in $G(\pi\beta)$ is nontrivial, then $\Delta c(\pi, \pi\beta)=2$ and $\Delta f(\pi, \pi\beta)=0$, so $\Delta g(\pi, \pi\beta)=-1$; 
 \item otherwise, $\Delta c_1(\pi, \pi\beta)=3$ and $\Delta f(\pi, \pi\beta)=-1$, which in turn implies $\Delta g(\pi, \pi\beta)= -1$. 
\end{enumerate}
\item we have three possible subcases:
 \begin{enumerate}
  \item $\Delta c_1(\pi, \pi\beta)=2$: the leftmost cycle in $G(\pi\beta)$ is nontrivial, so $\Delta f(\pi, \pi\beta)=0$ and $\Delta g(\pi, \pi\beta)=-1$; 
  \item $\Delta c_1(\pi, \pi\beta)=3$: the leftmost cycle in $G(\pi\beta)$ may or may not be trivial, so $\Delta f(\pi, \pi\beta)\le 0$ and $\Delta g(\pi, \pi\beta)\le -1$; 
  \item $\Delta c_1(\pi, \pi\beta)=4$: the leftmost cycle in $G(\pi\beta)$ is trivial, so $\Delta f(\pi, \pi\beta)=-1$ and $\Delta g(\pi, \pi\beta)= -2$.
 \end{enumerate}
 \end{enumerate}
\end{proof}

We handle the case where $\pi_1=1$ in the proof of our upper bound below.

\begin{theorem}\label{thm:first-upper-bound-on-pbid}
For any $\pi$ in $S_n$, we have $pbid(\pi)\leq g(\pi)$.
\end{theorem}
\begin{proof}
If $\pi_1\neq 1$, then we apply \autoref{lemma:pi-neq-1} to decrease $g(\pi)$ by at least $1$. Otherwise, 
$\{\pi'_0, \pi'_1\}$ is a $1$-cycle in $G(\pi)$ and $f(\pi)=0$.  
Assume $\pi\neq\iota$ to avoid triviality; then $G(\pi)$ contains a nontrivial cycle, from which we select a grey edge $\{\pi'_{2i-2}, \pi'_{2j-1}\}$ with $j>i$. Applying the prefix block-interchange $\beta(1, i, i, j)$ then makes $\pi_i$ and $\pi_{i}+1$ contiguous in $\pi\beta$, and that pair corresponds to a new $1$-cycle in $G(\pi\beta)$. On the other hand, $\beta$ merges the $1$-cycle $\{\pi'_0, \pi'_1\}$ in $G(\pi)$ with the cycle that contains $\{\pi'_{2i-2}, \pi'_{2j-1}\}$, 
so 
$\Delta c(\pi, \pi\beta)=0=\Delta c_1(\pi, \pi\beta)$, $\Delta f(\pi, \pi\beta)=1$ and \autoref{eqn:delta-g} yields $\Delta g(\pi, \pi\beta)=0/2-0-1=-1$.
\end{proof}

The smallest example of a permutation for which the inequality in \autoref{thm:first-upper-bound-on-pbid} is strict is $\pi=\langle 3\ 2\ 1\rangle$,  with $pbid(\pi)=1<g(\pi)=2$. \autoref{algo:2approx} implements the strategy described in \autoref{thm:first-upper-bound-on-pbid}. We prove in the next subsection that \autoref{algo:2approx} is a 2-approximation. 

\begin{algorithm}[htbp]
\caption{\textsc{ApproximateSbpbi}($\pi$)}
\label{algo:2approx}
\KwIn{A permutation $\pi$ of $[n]$.}
\KwOut{A sorting sequence of prefix block-interchanges for $\pi$.}
\BlankLine
$S\leftarrow$ empty sequence\;
\While{$\pi\neq\iota$}{
    \If{$\pi_1\neq 1$}{
        $j\leftarrow$ the position of $\pi_1-1$\;
        \If{there exists a pair $(\pi_i, \pi_k=\pi_i+1)$ such that $i\le j\le k$ and the corresponding grey arc belongs to the leftmost cycle of $G(\pi)$}{
            $\sigma\leftarrow \beta(1, i, j, k)$\tcp*{\autoref{lemma:pi-neq-1} cases 1--3}
        }
        \Else{
            $i, k\leftarrow$ positions such that $i\le j\le k$ and $\pi_k=\pi_i-1$\;
            $\sigma\leftarrow \beta(1, i, j, k)$\tcp*{\autoref{lemma:pi-neq-1} case 4}
        }
    }
    \Else(\tcp*[h]{\autoref{thm:first-upper-bound-on-pbid}}){
        $i\leftarrow$ smallest index such that $\pi_{i+1}\neq\pi_i+1$\;
        $j\leftarrow$ the position of $\pi_i+1$\;
        $\sigma\leftarrow \beta(1, i, i, j)$\;
    }
    $\pi\leftarrow \pi\sigma$\;
    $S$.append($\sigma$)\;
}
\textbf{return} $S$\;
\end{algorithm}

\subsection{A lower bound based on the breakpoint graph}\label{sec:lb}

We now prove a lower bound on $pbid$ that allows us to show that \autoref{algo:2approx} is a 2-approximation for \SBPBIshort. To that end, we use a framework introduced in \cite{labarre-lower}. 
The starting point is the following mapping, in which the symmetric group on $[n+1]$ is identified with the symmetric group on $\{0\}\cup[n]$ and where $A_n$ is the subgroup of $S_n$ formed by the set of all \emph{even} permutations, i.e. permutations with an even number of even cycles:
\begin{equation}\label{eq-alpha}
\psi:S_n\rightarrow A_{n+1}:\pi\mapsto\overline{\pi} =  (0,1,2,\ldots,n) (0,\pi_n,\pi_{n-1},\ldots,\pi_1).
\end{equation}
This mapping associates to every permutation $\pi$ another permutation $\overline{\pi}$ whose disjoint cycles are in one-to-one correspondence with the cycles of $G(\pi)$. As a result, terminology based on the disjoint cycle decomposition of $\overline{\pi}$ or on the alternating cycle decomposition of $G(\pi)$ can conveniently be used indifferently, including the notation introduced at the beginning of \autoref{sec:first-2-approx} (e.g. $c(\overline{\pi})=c(G(\pi))$, and therefore $\Delta c(\overline{\pi}, \overline{\pi\sigma})=c(\overline{\pi\sigma})-c(\overline{\pi})=c(G(\pi\sigma))-c(G(\pi))=\Delta c(\pi, \pi\sigma)$). The following result will be our main tool for proving our lower bound. 

\begin{theorem}\label{thm:main-theorem}\cite{labarre-lower}
Let $S$ be a subset of $S_n$ whose elements are mapped by $\psi(\cdot)$ onto $S'\subseteq A_{n+1}$. Moreover, let $\mathscr C$ be the union of the conjugacy classes (of $S_{n+1}$) that intersect with $S'$; then for any $\pi$ in $S_n$, any factorisation of $\pi$ into $t$ elements of $S$ yields a factorisation of $\overline{\pi}$ into $t$ elements of $\mathscr C$.
\end{theorem}

Consequently, if we let $d_S(\sigma)$ denote the length of a shortest sorting sequence for $\sigma$ consisting solely of elements from $S$, then 
\autoref{thm:main-theorem} implies that for any $\pi$ in $S_n$ and any choice of $S\subseteq S_n$, we have $d_S(\pi)\ge d_{S'}(\overline{\pi})$. 
In order to use \autoref{thm:main-theorem}, we need a translation of the effect of an operation on $\pi$ in terms of a transformation on $\overline{\pi}$, as well as a precise characterisation of the image of a prefix block-interchange under the mapping $\psi$. Both are provided, respectively, by the following results.

\begin{lemma}\label{lemma:value-of-overline-brack-pi-circ-sigma-brack}
\cite{labarre-lower} For all $\pi$, $\sigma$ in $S_n$, we have $\overline{\pi\sigma}=\overline{\pi}(\overline{\sigma}^\pi)$.
\end{lemma}

\begin{lemma}\label{lemma:block-interchange-image}\cite{labarre-lower}
For any block-interchange $\beta(i, j, k, \ell)$ in $S_n$, we have $\overline{\beta(i, j, k, \ell)}=(j, \ell)(i, k)$.
\end{lemma}

As is well-known, a $2$-cycle in a permutation $\sigma$ containing elements from different cycles in a permutation $\pi$ merges those cycles in $\pi\sigma$, while a $2$-cycle in $\sigma$ containing elements from the same cycle in $\pi$ splits that cycle into two cycles in $\pi\sigma$. \autoref{lemma:value-of-overline-brack-pi-circ-sigma-brack} and \autoref{lemma:block-interchange-image} therefore provide us with a very simple way of analysing the effects of a block-interchange: the effect of $\beta$ on the cycles of $G(\pi)$ is the same as the effect of $\overline{\beta}^\pi$ on the cycles of $\overline{\pi}$, and therefore bounds on the (prefix) block-interchange distance of $\pi$ can be obtained by studying the effects of pairs of $2$-cycles on $\overline{\pi}$. 
The following lemma will be useful in restricting the number of cases in the proof of our lower bound (\autoref{thm:first-lower-bound-on-pbid}).

\begin{lemma}\label{lemma:values-for-delta-c}
 For any $\pi$ in $S_n$ and any block-interchange $\beta$, we have $\Delta c(\pi,\pi\beta)\in\{-2, 0, 2\}$.
\end{lemma}
\begin{proof}
 By \autoref{lemma:block-interchange-image}, $\overline{\beta}$ consists of two $2$-cycles, each of which might split a cycle into two cycles or merge two cycles into one (\autoref{lemma:value-of-overline-brack-pi-circ-sigma-brack}). Combining all possible cases yields the set $\{-2, 0, 2\}$ as possible values for $\Delta c(\overline{\pi}, \overline{\pi\beta})=\Delta c(\pi,\pi\beta)$.
\end{proof}

Finally, the following technical observation will be useful in ruling out impossible values for $\Delta f(\pi, \sigma)$, whose set of possible values is $\{-1, 0, 1\}$ when no restrictions apply.

\begin{lemma}\label{obs:deltac1-implications-deltaf}
For any $\pi$ in $S_n$ and every prefix block-interchange $\beta$: if $\Delta c_1(\pi,\pi\beta)\ge 2$, then $\Delta f(\pi,\pi\beta)\neq 1$.
\end{lemma}
\begin{proof}
If $\Delta c_1(\pi,\pi\beta)\ge 2$, then the new $1$-cycles are obtained in one of the following ways:
\begin{enumerate}
 \item if at least one of them is the result of a split of the leftmost cycle of $G(\pi)$, then that cycle is nontrivial and therefore $f(\pi)=1$, thereby forbidding the value $\Delta f(\pi,\pi\beta)=1$;
 \item otherwise, all new $1$-cycles are extracted from a cycle in $G(\pi)$ other than the leftmost cycle; since that cycle can only be split into at most two new cycles (\autoref{lemma:value-of-overline-brack-pi-circ-sigma-brack} and \autoref{lemma:block-interchange-image}), we have $\Delta c_1(\pi,\pi\beta)\le 2$ in this case. Moreover, we also have $\pi_1=1$, otherwise the $1$-cycle containing $\pi_1$ would vanish in $G(\pi\beta)$ and contradict our assumption that $\Delta c_1(\pi,\pi\beta)\ge 2$. Therefore, the value $\Delta f(\pi,\pi\beta)=1$ is also excluded in this case.
\end{enumerate}
\end{proof}

We now have everything we need to prove our lower bound on $pbid$. 

\begin{theorem}\label{thm:first-lower-bound-on-pbid}
For any $\pi$ in $S_n$, we have $pbid(\pi)\geq g(\pi)/2$.
\end{theorem}
\begin{proof}
By \autoref{thm:main-theorem} and \autoref{lemma:block-interchange-image}, we have $pbid(\pi)\ge d(\overline{\pi})$, where $d(\overline{\pi})$ is the length of a shortest sorting sequence for $\overline{\pi}$ where the only nontrivial cycles of each transformation in the sequence are two $2$-cycles, exactly one of which contains $1$. As a result, any lower bound on $d(\overline{\pi})$ is a lower bound on $pbid(\pi)$, and therefore we only need to show that a transformation of the kind we have just described can decrease the value of $g(\pi)$ by at most $2$.
 
Let $\overline{\beta}=(1, a)(b, c)$ be the image of a prefix block-interchange under the mapping $\psi(\cdot)$.  By \autoref{lemma:values-for-delta-c}, we only need to distinguish between the following three cases; in each situation, we aim to minimise the value of $\Delta g(\overline{\pi}, \overline{\pi\beta})$.
 
 \begin{enumerate}
  \item If $\Delta c(\overline{\pi}, \overline{\pi\beta})=-2$, then clearly $\Delta c_1(\overline{\pi}, \overline{\pi\beta})\le 0$, 
  and \autoref{eqn:delta-g} allows us to conclude that $\Delta g(\overline{\pi}, \overline{\pi\beta})\ge -1-0-1=-2$.
  
  \item If $\Delta c(\overline{\pi}, \overline{\pi\beta})=0$, then either $2$-cycle of $\beta$ merges two cycles while the other splits a cycle into two. The lengths of the involved cycles in $\overline{\pi}$ and in $\overline{\pi\beta}$ may vary, but this observation is enough to deduce that $\Delta c_1(\overline{\pi}, \overline{\pi\beta}) \le 2$.
The lowest value of $\Delta g(\overline{\pi}, \overline{\pi\beta})$ is obtained when $\Delta c_1(\overline{\pi}, \overline{\pi\beta})=2$, in which case \autoref{eqn:delta-g} and \autoref{obs:deltac1-implications-deltaf} yield $\Delta g(\overline{\pi}, \overline{\pi\beta})\ge 0-2-0=-2$, or when $\Delta c_1(\overline{\pi}, \overline{\pi\beta})=1$, in which case \autoref{eqn:delta-g} 
yields $\Delta g(\overline{\pi}, \overline{\pi\beta})\ge 0-1-1=-2$.

  \item If $\Delta c(\overline{\pi}, \overline{\pi\beta})=2$, then both elements of $\beta$ each split one cycle into two cycles. As in the previous case, the lengths of the involved cycles in $\overline{\pi}$ and in $\overline{\pi\beta}$ may vary, but this observation is enough to deduce that $\Delta c_1(\overline{\pi}, \overline{\pi\beta}) \le 4$, and as a result $\Delta f(\overline{\pi}, \overline{\pi\beta}) \in \{-1, 0\}$ (\autoref{obs:deltac1-implications-deltaf}). The lowest value of $\Delta g(\overline{\pi}, \overline{\pi\beta})$ is obtained in two cases:
  \begin{enumerate}
  \item when $\Delta c_1(\overline{\pi}, \overline{\pi\beta})=4$, in which case the leftmost cycle of $\overline{\pi}$ splits into two $1$-cycles; therefore $\Delta f(\overline{\pi}, \overline{\pi\beta})=-1$ and \autoref{eqn:delta-g} yields $\Delta g(\overline{\pi}, \overline{\pi\beta})\ge 1-4+1=-2$;
  \item or when $\Delta c_1(\overline{\pi}, \overline{\pi\beta})=3$, in which case \autoref{eqn:delta-g} and \autoref{obs:deltac1-implications-deltaf} yield $\Delta g(\overline{\pi}, \overline{\pi\beta})\ge 1-3+0=-2$.
  \end{enumerate}
  \end{enumerate}
\end{proof}

\autoref{thm:first-lower-bound-on-pbid} implies that \autoref{algo:2approx} is a 2-approximation for \SBPBIshort. 

\section{Tightening the bounds}\label{sec:tighter}

Although obtaining better approximation guarantees for {\SBPBIshort} seems as nontrivial as for other prefix sorting problems, the bounds obtained in the previous section can be improved. We show in this section how to tighten them, and then use those improved results in \autoref{sec:diameter} to compute the maximal value that the distance can reach.

\subsection{A tighter upper bound}

By \autoref{thm:first-lower-bound-on-pbid}, the largest value by which the upper bound of \autoref{thm:first-upper-bound-on-pbid} can decrease with a single prefix block-interchange is $2$. In this section, we characterise all permutations which admit such a prefix block-interchange. Other nontight permutations exist (see e.g. \autoref{lemma:2-cycle-except-LMC-implies-better-move-possible}), but they do not admit such an operation as the first step of an optimal sorting sequence. As a consequence, we obtain an improved upper bound on $pbid$ in \autoref{thm:second-upper-bound-on-pbid}. 

\begin{lemma}\label{lemma:2-cycle-crosses-first-inner-grey-edge}
 For any $\pi$ in $S_n$: if $G(\pi)$ contains a $2$-cycle that intersects the first grey edge, 
then there exists a prefix block-interchange $\beta$ such that $\Delta g(\pi, \pi\beta)=-2$.
\end{lemma}
\begin{proof}
Follows from cases 4b and 4c of the proof of \autoref{lemma:pi-neq-1}, when the cycle that contains grey edge $f$ has length $2$.
\end{proof}

Following \cite{bafna-transpositions}, we say that a cycle $C$ with $b_i$ and $b_k$ as black edges of minimum and maximum indices, respectively, \emph{spans} a black edge $b_j$ if $i<j<k$.

\begin{lemma}\label{lemma:2-cycle-spans-edge-outside-LMC}
 For any $\pi$ in $S_n$: if $G(\pi)$ contains a $2$-cycle which is not the leftmost cycle and which spans a black edge that belongs to a nontrivial cycle different from the leftmost cycle, then there exists a prefix block-interchange $\beta$ such that $\Delta g(\pi, \pi\beta)=-2$. 
\end{lemma}
\begin{proof}
We apply a prefix block-interchange defined by the first black edge, both black edges of the $2$-cycle, and any black edge spanned by the $2$-cycle:
 
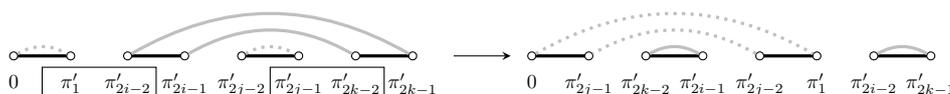
\begin{center}
\begin{tikzpicture}[scale=.75,transform shape,baseline]

    \foreach [count=\i] \color in {black,black,black,black}
        \draw[color=\color,very thick] (2*\i-2, 0) -- (2*\i-1, 0);

    \foreach [count=\i from 0] \name in {a,b,c,d,e,f,g,h}
    \node[vertex] (\name) at (\i,0) [draw,circle] 
    {};
    \foreach \name/\lab in {a/$0$,b/$\pi'_1$,c/$\pi'_{2i-2}$,d/$\pi'_{2i-1}$,e/$\pi'_{2j-2}$,f/$\pi'_{2j-1}$,g/$\pi'_{2k-2}$,h/$\pi'_{2k-1}$}
    \node at (\name) [label=below:\strut\lab] {};


    \begin{scope}[color=gray!50, very thick]
        \draw[] (d) to [bend left] (g);
        \draw[dotted] (a) to [bend left] (b);
        \draw[dotted] (e) to [bend left] (f);
        \draw[] (c) to [bend left] (h);
        
    \end{scope}
    
    \draw (.5, -.75) rectangle (2.5, -.2);
    \draw (4.5, -.75) rectangle (6.5, -.2);
\end{tikzpicture} 
\begin{tikzpicture}[scale=.75,transform shape,baseline]
\draw[->, >=stealth] (0, 0) -- (1, 0);
\end{tikzpicture} 
\begin{tikzpicture}[scale=.75,transform shape,baseline]

    \foreach [count=\i] \color in {black,black,black,black}
        \draw[color=\color,very thick] (2*\i-2, 0) -- (2*\i-1, 0);

    \foreach [count=\i from 0] \name in {a,f,g,d,e,b,c,h}
    \node[vertex] (\name) at (\i,0) [draw,circle] 
    {};
    \foreach \name/\lab in {a/$0$,b/$\pi'_1$,c/$\pi'_{2i-2}$,d/$\pi'_{2i-1}$,e/$\pi'_{2j-2}$,f/$\pi'_{2j-1}$,g/$\pi'_{2k-2}$,h/$\pi'_{2k-1}$}
    \node at (\name) [label=below:\strut\lab] {};


    \begin{scope}[color=gray!50, very thick]
        \draw[] (d) to [bend right] (g);
        \draw[dotted] (a) to [bend left] (b);
        \draw[dotted] (e) to [bend right] (f);
        \draw[] (c) to [bend left] (h);
        
    \end{scope}

\end{tikzpicture} 
\end{center}
The number of cycles does not change, so $\Delta c(\pi,\pi\beta)=0$. Either $\pi_1=1$, and then $\Delta c_1(\pi,\pi\beta)=1$ and $\Delta f(\pi,\pi\beta)=1$; or $\pi_1\neq 1$, and then  $\Delta c_1(\pi,\pi\beta)=2$ and $\Delta f(\pi,\pi\beta)=0$. In both cases,  \autoref{eqn:delta-g} yields $\Delta g(\pi,\pi\beta)=-2$.
\end{proof}

$2$-cycles other than the leftmost cycle and in a different configuration from our characterisations are still helpful. We show that even though they do not allow a prefix block-interchange that decreases $g(\cdot)$ by $2$ right away, they make it possible to obtain such an operation  \emph{eventually}. 

\begin{proposition}\label{lemma:2-cycle-except-LMC-implies-better-move-possible}
 For any $\pi$ in $S_n$: if $G(\pi)$ contains a $2$-cycle which is not the leftmost cycle, then $\pi$ admits a sequence $S$  of prefix-block interchanges that turns $\pi$ into a permutation $\sigma$ with $\Delta g(\pi, \sigma)=|S|$ and which admits a prefix block-interchange $\beta$ such that $\Delta g(\sigma, \sigma\beta)=-2$.
\end{proposition}
\begin{proof}
Let $C$ denote the $2$-cycle of interest. If $C$ intersects the first grey edge or a cycle different from the leftmost cycle, then we are done (see respectively \autoref{lemma:2-cycle-crosses-first-inner-grey-edge} and \autoref{lemma:2-cycle-spans-edge-outside-LMC}). Otherwise, $C$ intersects another grey edge of the leftmost cycle, and \autoref{lemma:pi-neq-1} allows us to reduce $g(\pi)$ by one while reducing the length of the leftmost cycle without affecting $C$. Repeated applications of  \autoref{lemma:pi-neq-1} eventually yield a permutation $\sigma$ which satisfies one of the following conditions:
\begin{enumerate}
 \item $\sigma_1=1$, in which case $C$ necessarily spans a black edge that does not belong to the leftmost cycle and therefore we can apply \autoref{lemma:2-cycle-spans-edge-outside-LMC};
 \item $\sigma_1\neq 1$ and $C$ intersects another cycle than the leftmost cycle, in which case we can again apply \autoref{lemma:2-cycle-spans-edge-outside-LMC}; or
 \item $\sigma_1\neq 1$ and $C$ intersects the first grey edge, in which case we can apply \autoref{lemma:2-cycle-crosses-first-inner-grey-edge}.
\end{enumerate}
\end{proof}

The above results allow us to easily identify other nontight permutations (with respect to \autoref{thm:first-upper-bound-on-pbid}) whose breakpoint graph contains no $2$-cycle. For instance, if the first grey edge intersects a $3$-cycle $C$, then applying a prefix-block interchange selected according to \autoref{lemma:2-cycle-crosses-first-inner-grey-edge} decreases  the lengths of both the leftmost cycle and $C$, which becomes a $2$-cycle and which therefore eventually allows for a prefix block-interchange that decreases $g(\cdot)$ by $2$ according to \autoref{lemma:2-cycle-except-LMC-implies-better-move-possible}.

The interactions between $2$-cycles prevent us from simply reducing the upper bound of \autoref{thm:first-upper-bound-on-pbid} by the number of $2$-cycles in $G(\pi)$: indeed, the black edge spanned by the $2$-cycle described in \autoref{lemma:2-cycle-spans-edge-outside-LMC} may belong to a $2$-cycle whose length will increase in the resulting permutation. Therefore, we can only conclude the following.

\begin{theorem}\label{thm:second-upper-bound-on-pbid}
For any $\pi$ in $S_n$, we have $pbid(\pi)\leq g(\pi)-\lceil c_2^{\emptyset}(G(\pi))/2\rceil$, where $c_2^{\emptyset}(G(\pi))$ denotes the number of $2$-cycles in $G(\pi)$ excluding the leftmost cycle.
\end{theorem}
\begin{proof}
 We repeatedly apply \autoref{lemma:2-cycle-except-LMC-implies-better-move-possible} to take advantage of suitable $2$-cycles. Each prefix block-interchange we use transforms a $2$-cycle into two $1$-cycles without affecting the other $2$-cycles, except possibly in the case of \autoref{lemma:2-cycle-spans-edge-outside-LMC} when the edge spanned by the $2$-cycle we focus on belongs to another $2$-cycle. In the worst case, every $2$-cycle we try to split forces us to increase the length of a $2$-cycle it intersects, hence the improvement of only $\lceil c_2^{\emptyset}(G(\pi))/2\rceil$ over \autoref{thm:first-upper-bound-on-pbid}.
\end{proof}

\autoref{thm:second-upper-bound-on-pbid} again yields a tight upper bound, as shown by the permutation $\langle 1\ 4\ 3\ 2\rangle$ for which the value of the improved upper bound matches its distance.

\subsection{A tighter lower bound}

A trivial lower bound on $pbid$ is given by the value of the block-interchange distance (denoted by $bid(\pi)$), which can be computed in $O(n)$ time thanks to the following result.

\begin{theorem}\label{thm:bid}
 \cite{Christie1996} For any $\pi$ in $S_n$, we have $bid(\pi)=(n+1-c(G(\pi)))/2$.
\end{theorem}
 
This lower bound often outperforms that of \autoref{thm:first-lower-bound-on-pbid}, but cases exist where the opposite holds ($\langle 1\ 4\ 3\ 2 \rangle$ is the smallest example). As we show below, it is possible to build on this trivial lower bound to obtain a much better and useful lower bound. The resulting lower bound also allows us to compute the maximum value that the prefix block-interchange can reach, a problem we address in \autoref{sec:diameter}.

\begin{definition}
\cite{DBLP:journals/jacm/HannenhalliP99} Let $\pi$ be a permutation. Two cycles $C$ and $D$ of $G(\pi)$ \emph{intersect} if $C$ contains a grey edge $e$ that intersects with a grey edge $f$ of $D$. A \emph{component} of $G(\pi)$ is a connected component of the intersection graph of the nontrivial cycles of $G(\pi)$.
\end{definition}

For instance, the breakpoint graph of \autoref{fig:breakpoint-graph} page~\pageref{fig:breakpoint-graph} has two components: the leftmost cycle, and the pair of intersecting $2$-cycles. Let $CC(G(\pi))$ denote the number of components of $G(\pi)$. We first show that prefix block-interchanges that merge components of the breakpoint graph cannot decrease the number of cycles it contains\footnote{See \autoref{app:proof-lemma:merging-components-cannot-decrease-bid} for the proof.}.

\begin{lemma}\label{lemma:merging-components-cannot-decrease-bid}
 For any $\pi$ in $S_n$, let $\beta$ be a prefix block-interchange with $CC(G(\pi\beta))<CC(G(\pi))$; then $\Delta c(G(\pi, \pi\beta))\in \{-2, 0\}$.
\end{lemma}

\begin{theorem}\label{thm:second-lower-bound-on-pbid}
 For any $\pi$ in $S_n$, we have $pbid(\pi)\ge bid(\pi) + CC(G(\pi)) -  \left\{\begin{array}{ll}
0 & \mbox{if } \pi_1=1, \\
1 & \mbox{otherwise}.
\end{array}
\right.
$
\end{theorem}
\begin{proof}
The expression for the lower bound corresponds to the following strategy: for each component $C$ of $G(\pi)$, sort the corresponding subpermutation if $C$ contains $0$, or use a prefix block-interchange to make it contain $0$ and then sort it. Trivially, the number of steps in the sorting stage cannot be lower than the number or unrestricted block-interchanges it would require. Any other strategy would have to merge components; however, by \autoref{lemma:merging-components-cannot-decrease-bid}, a prefix block-interchange $\beta$ that merges connected components cannot increase the number of cycles, and therefore $bid(\pi\beta)\ge bid(\pi)$ for any permutation $\pi$ and any such prefix block-interchange.
\end{proof}

\section{The maximum value of the prefix block-interchange distance}\label{sec:diameter}

The \emph{diameter} of $S_n$ is the maximum value that a distance can reach for a particular family of operations. 
In this section, we use our results to compute its exact value in the case of prefix block-interchanges, and show along the way that our 2-approximation algorithm based on the breakpoint graph is also a 2-approximation with respect to the following notion.

\begin{definition}\cite{dias-prefix}\label{def:breakpoint}
Let $\pi$ be a permutation of $\{0, 1, 2, \ldots, n+1\}$ with $\pi_0=0$ and $\pi_{n+1}=n+1$. The pair $(\pi_i, \pi_{i+1})$ with $0\le i\le n$ is a \emph{breakpoint} if $i=0$ or $\pi_{i+1} - \pi_i\neq 1$, and an \emph{adjacency} otherwise. The number of breakpoints in a permutation $\pi$ is denoted by $\breakpoints{\pi}$. 
\end{definition}

For readability, we slightly abuse notation by using $\breakpoints{\pi}$ for $\pi$ in $S_n$, with the understanding that it refers to $\breakpoints{\langle 0\ \pi_1\ \pi_2\ \cdots\ \pi_n\ n+1\rangle}$. We let $\Delta b(\pi, \sigma)=b(\sigma)-b(\pi)$, and say that a prefix block-interchange $\beta$ with $\Delta b(\pi, \pi\beta)<0$ \emph{removes} breakpoints, or \emph{creates} adjacencies.

\begin{lemma}\label{lemma:delta-breakpoints-pbi}
For any $\pi$ in $S_n$ and any prefix block-interchange $\beta$, we have $|\Delta b(\pi,\pi\beta)|\le 3$.
\end{lemma}
\begin{proof}
A prefix block-interchange $\beta$ acts on at most four pairs of adjacent elements, including the pair $(0,\pi_1)$ which always counts as a breakpoint. Therefore, the number of breakpoints that $\beta$ can remove or create lies in the set $\{0, 1, 2, 3\}$.
\end{proof}

Since $\iota$ is the only permutation with exactly one breakpoint, \autoref{lemma:delta-breakpoints-pbi} immediately implies the following corollary.

\begin{corollary}\label{lemma:breakpoint-lower-bound-on-pbid}
For any $\pi$ in $S_n: pbid(\pi)\geq \left\lceil\frac{\breakpoints{\pi}-1}{3}\right\rceil.$
\end{corollary}

\begin{lemma}\label{lemma:breakpoint-upper-bound-on-pbid}
 For any $\pi$ in $S_n$, we have $pbid(\pi)\le 2\left\lceil\frac{\breakpoints{\pi}-1}{3}\right\rceil$.
\end{lemma}
\begin{proof}
Assume $\pi\neq\iota$ to avoid triviality, and observe that adjacencies in $\langle 0\ \pi_1\ \pi_2\ \cdots\ \pi_n\ n+1\rangle$ are in one-to-one correspondence with trivial cycles in $G(\pi)$ (except for the pair $(0, \pi_1)$ which by \autoref{def:breakpoint} is always a breakpoint). If $\pi_1\neq 1$, then \autoref{lemma:pi-neq-1} guarantees the existence of a prefix block-interchange $\beta$ with $\Delta c_1(\pi, \pi\beta)\ge 2$ and in turn implies $\Delta b(\pi, \pi\beta)\ge 2$. If $\pi_1=1$, then we select $\beta$ as in the proof of \autoref{thm:first-upper-bound-on-pbid}, which creates a new trivial cycle in $G(\pi\beta)$ that corresponds to a new adjacency in $\pi\beta$. Since $\pi\beta_1\neq 1$, the previous case provides the next operation, and the number of breakpoints decreases by at least three using two prefix block-interchanges.
\end{proof}

Since $\breakpoints{\pi}\le n+1$ for all $\pi$ in $S_n$, we immediately obtain the following.

\begin{corollary}\label{cor:upper-bound-diameter}
 For any $\pi\in S_n$, we have $pbid(\pi)\le 2n/3$.
\end{corollary}

\begin{theorem}
 The diameter of $S_n$ under prefix block-interchanges is $\lfloor 2n/3\rfloor$.
\end{theorem}
\begin{proof}
The cases where $n\le 2$ are easily verified. 
We build tight families of permutations for any $n\ge 3$, starting with permutations $\pi=\langle 1\ 3\ 2\rangle$, $\sigma=\langle 1\ 4\ 3\ 2\rangle$, and $\tau=\langle 1\ 3\ 2\ 5\ 4\rangle$ as base cases for the values of $n$, $n-1$ and $n-2$ that are multiples of $3$, respectively. \autoref{thm:second-lower-bound-on-pbid} yields $pbid(\pi)\ge 2$, $pbid(\sigma)\ge 2$ and $pbid(\tau)\ge 3$, while \autoref{cor:upper-bound-diameter} yields $pbid(\pi)\le 2$, $pbid(\sigma)<3$ and $pbid(\tau)< 4$, thereby matching the lower bounds.

To obtain tight permutations for larger values of $n$, we concatenate the sequence $\langle n+1\ n+3\ n+2\rangle$ to $\pi$, $\sigma$ or $\tau$, and repeat the process as many times as needed. Each concatenation preserves the congruence of $n$ and adds a new component to $G(\cdot)$ which consists of an isolated $3$-cycle. The lower bound of \autoref{thm:second-lower-bound-on-pbid} thereby increases by $2$ with each concatenation, as does the upper bound of \autoref{cor:upper-bound-diameter}. As a result, a permutation with prefix block-interchange distance $\lfloor 2n/3\rfloor$ exists for every value of $n$ in $\mathbb{N}$.
\end{proof}

While many permutations reach the diameter when $n\not\equiv 0\pmod{3}$, the permutation $\langle 1\ 3\ 2\ 4\ 6\ 5\ \cdots\ n-2\ n\ n-1\rangle$ seems to be the only tight permutation when $n\equiv 0\pmod{3}$.

\section{Conclusions and future work}

We initiated in this work the study of sorting permutations by prefix block-interchanges, an operation that generalises several well-studied operations in genome rearrangements and interconnection network design. We gave tight upper and lower bounds on the corresponding distance, and derived a 2-approximation algorithm for the problem. We then showed how to obtain better bounds on the distance using a finer analysis of cycles and components of the breakpoint graph, and determined the maximum value that the distance can reach.

Several questions remain open, most notably the complexity of \SBPBIshort, and its approximability if it turns out to be \NP-complete. We note that improving the ratio of 2 will require improved lower bounds, since for all three upper bounds we have obtained (\autoref{thm:first-upper-bound-on-pbid}, \autoref{thm:second-upper-bound-on-pbid} and \autoref{lemma:breakpoint-upper-bound-on-pbid}) there are permutations whose actual distance match those bounds. 
A number of leads seem promising in that regard, the most obvious one being the computation of the exact value of the ``special purpose distance'' introduced in the proof of \autoref{thm:first-lower-bound-on-pbid}, as well as a more intricate analysis of the cycles of the breakpoint graph as well as their interactions as initiated in \autoref{sec:tighter}. 
Given how helpful $2$-cycles are in decreasing the upper bound of \autoref{thm:first-upper-bound-on-pbid}, it would seem natural to focus on \emph{simple permutations} (i.e. permutations whose breakpoint graph contains no cycle of length $>2$). This strategy eventually led to a polynomial-time algorithm for sorting by signed reversals~\cite{DBLP:journals/jacm/HannenhalliP99}, but we do not expect such an outcome for prefix block-interchanges since the 
simplification process does not preserve the prefix block-interchange distance (whereas it did preserve the signed reversal distance): the smallest counterexample is $\pi=\langle 3\ 1\ 4\ 2\rangle$, which simplifies to $\sigma=\langle 5\ 2\ 7\ 4\ 1\ 6\ 3\rangle$, and for which $pbid(\pi)=2\neq pbid(\sigma)=3$. 

In a broader context, we also hope that our results and the strategies we designed to tackle \SBPBIshort\ can be applied to other prefix sorting problems (for instance, a generalisation of the lower bounding strategy of \autoref{thm:second-lower-bound-on-pbid} to any distance would be of interest). The breakpoint graph approach provides a clear strategy for unrestricted sorting problems, which, informally, usually consists in increasing the number of cycles in as few steps as possible. As our bounds show, and as has been observed for most prefix sorting problems~\cite{akers-star,labarre-lower,labarre-burnt-pancakes}, this no longer works under the prefix constraint since operations that \emph{decrease} or do not affect the number of cycles can also decrease the value of our bounds. Nevertheless, bounds obtained for prefix exchanges, prefix block-transpositions, prefix block-interchanges and prefix signed reversals are all based on $g(\cdot)$, which seems to indicate common underlying features that could be taken advantage of, and possibly lead to a common framework for approximating these problems or solving them exactly.

\bibliographystyle{plainurl}
\bibliography{sbpbi}

\clearpage
\appendix
\section{Proof of \autoref{lemma:merging-components-cannot-decrease-bid}}\label{app:proof-lemma:merging-components-cannot-decrease-bid}

A prefix block-interchange $\beta$ acts on one, two, three or four cycles. We first show that if $\beta$ reduces the number of connected components of $G(\pi)$, then it cannot act on a single cycle of $G(\pi)$. This is important for the proof of \autoref{lemma:merging-components-cannot-decrease-bid}, because some prefix block-interchanges acting on a single cycle increase the number of cycles and therefore may decrease the value of $bid(\pi)$ regardless of their effect on $pbid(\pi)$ or $g(\pi)$. The following concepts will be helpful.

\begin{definition}
 For any permutation $\pi$, let $e=(e_1, e_2)$ and $f=(f_1, f_2)$ be two grey edges in $G(\pi)$, with $e_1<e_2$ and $f_1<f_2$. We say that $e$ and $f$ are \emph{independent} if they are:
 \begin{itemize}
  \item \emph{nested}, i.e. $e_1<f_1<f_2<e_2$ (written $f\subset_\pi e$) or $f_1<e_1<e_2<f_2$ (written $e\subset_\pi f$); or
  \item \emph{ordered}, i.e. $e_1<e_2<f_1<f_2$, in which case we say that $e$ \emph{precedes} $f$ (written $e<_\pi f$), or $f_1<f_2<e_1<e_2$ (i.e. $f$ precedes $e$).
 \end{itemize}
\end{definition}

Grey edges naturally define intervals in $\pi'$, so we use the same notation to compare intervals, or grey edges with intervals.  We will sometimes need to distinguish  \emph{proper} block-interchanges, i.e. of the form $\beta(i, j, k, \ell)$ with $j<k$, from prefix block-transpositions, which are of the form $\beta(i, j, j, \ell)$.

\begin{lemma}\label{lemma:if-CC-decreases-then-pbi-acts-on-more-than-one-cycle}
  For any $\pi$ in $S_n$, let $\beta$ be a prefix block-interchange with $CC(G(\pi\beta))<CC(G(\pi))$; then $\beta$ cannot act on a single cycle of $G(\pi)$.
\end{lemma}
\begin{proof}
Let $e=(e_1, e_2)$ and $f=(f_1, f_2)$ with $e_1<e_2$ and $f_1<f_2$ be two grey edges of $G(\pi)$. We show that if $e$ and $f$ are independent in $G(\pi)$, then they remain independent in $G(\pi\beta)$. For readability, we assume that the indices of $e$ and $f$ correspond to positions in $\pi$ rather than $\pi'$. 
The connections between the black edges of $G(\pi)$ on which $\beta(1, i, j, k)$ acts imply the following:
\begin{itemize}
 \item both $e_1$ and $f_1$ lie in the interval $[1, k]$, otherwise $\beta$ would not affect $e$ or $f$;
 \item at least $e$ or $f$ has both endpoints in $[1,i]$, $[i, j]$ (which is empty if $\beta$ is not proper) or $[j,k]$, otherwise they both intersect the cycle on which $\beta$ acts and therefore $CC(G(\pi\beta))\ge CC(G(\pi))$;
 \item if $x \subset_\pi I$ holds for all possible combinations of $x$ in $\{e, f\}$ and $I$ in $\{[1, i], [i, j], [j, k]\}$, then $\beta$ trivially preserves the order of endpoints and therefore the interactions between $e$ and $f$ as well, so that their independence in $G(\pi)$ is preserved in $G(\pi\beta)$.
\end{itemize}
Without loss of generality, the only cases left to examine are those where $e\subset_\pi f$ and $e\subset_\pi [1, i]$, $e\subset_\pi [i, j]$, or $e\subset_\pi [j, k]$. The only two ways of making $e$ and $f$ intersect in $G(\pi\beta)$ are therefore either to exchange $e_2$ and $f_2$ without moving $e_1$ and $f_1$, which is impossible because $\beta$ is a \emph{prefix} block-interchange; or to exchange $f_1$ and $e_1$ without moving $e_2$ and $f_2$, which is impossible as well since $\beta$ must act on the four black edges of the cycle.
\end{proof}

We can now prove \autoref{lemma:merging-components-cannot-decrease-bid}.

\begin{proof}[Proof of \autoref{lemma:merging-components-cannot-decrease-bid}]
By \autoref{lemma:if-CC-decreases-then-pbi-acts-on-more-than-one-cycle}, we have the following three cases to analyse:
 
 \begin{enumerate}
  \item if $\beta$ acts on two cycles from different components, then two or three of the black edges on which $\beta$ acts belong to the same cycle. In all resulting cases, we have $\Delta c(\pi, \pi\beta)=0$; omitted cases are symmetric, and only proper block-interchanges are considered since the property we seek to prove is already known to hold for block-transpositions (see \cite[Lemma 3.2 page 228]{bafna-transpositions}):

  \begin{enumerate}
  \item 
  
  \begin{tikzpicture}[scale=.75,transform shape,baseline]

    \foreach [count=\i] \color in {black,black,black,black}
        \draw[color=\color,very thick] (2*\i-2, 0) -- (2*\i-1, 0);

    \foreach [count=\i from 0] \name in {a,b,c,d,e,f,g,h}
    \node[vertex] (\name) at (\i,0) [draw,circle] 
    {};
    \foreach \name/\lab in {a/$0$,b/$\pi'_1$,c/$\pi'_{2i-2}$,d/$\pi'_{2i-1}$,e/$\pi'_{2j-2}$,f/$\pi'_{2j-1}$,g/$\pi'_{2k-2}$,h/$\pi'_{2k-1}$}
    \node (\name2) at (\name) [label=below:\strut\lab] {};


    \begin{scope}[color=gray!50, very thick, dotted]
        \foreach \u/\v in {a/d, b/c, e/h, f/g}
            \draw[] (\u) to [bend left] (\v);
    \end{scope}
    
    \draw (.5, -.75) rectangle (2.45, -.2);
    \draw (4.5, -.75) rectangle (6.5, -.2);
\end{tikzpicture} 
\begin{tikzpicture}[scale=.75,transform shape,baseline]
\draw[->, >=stealth] (0, 0) -- (1, 0);
\end{tikzpicture} 
\begin{tikzpicture}[scale=.75,transform shape,baseline]

    \foreach [count=\i] \color in {black,black,black,black}
        \draw[color=\color,very thick] (2*\i-2, 0) -- (2*\i-1, 0);

    \foreach [count=\i from 0] \name in {a,f,g,d,e,b,c,h}
    \node[vertex] (\name) at (\i,0) [draw,circle] 
    {};
    \foreach \name/\lab in {a/$0$,b/$\pi'_1$,c/$\pi'_{2i-2}$,d/$\pi'_{2i-1}$,e/$\pi'_{2j-2}$,f/$\pi'_{2j-1}$,g/$\pi'_{2k-2}$,h/$\pi'_{2k-1}$}
    \node at (\name) [label=below:\strut\lab] {};

    \begin{scope}[color=gray!50, very thick, dotted]
        \foreach \u/\v in {a/d, b/c, e/h, f/g}
            \draw[] (\u) to [bend left] (\v);
    \end{scope}
\end{tikzpicture} 

\item 

  \begin{tikzpicture}[scale=.75,transform shape,baseline]

    \foreach [count=\i] \color in {black,black,black,black}
        \draw[color=\color,very thick] (2*\i-2, 0) -- (2*\i-1, 0);

    \foreach [count=\i from 0] \name in {a,b,c,d,e,f,g,h}
    \node[vertex] (\name) at (\i,0) [draw,circle] 
    {};
    \foreach \name/\lab in {a/$0$,b/$\pi'_1$,c/$\pi'_{2i-2}$,d/$\pi'_{2i-1}$,e/$\pi'_{2j-2}$,f/$\pi'_{2j-1}$,g/$\pi'_{2k-2}$,h/$\pi'_{2k-1}$}
    \node (\name2) at (\name) [label=below:\strut\lab] {};


    \begin{scope}[color=gray!50, very thick, dotted]
        \foreach \u/\v in {a/h, b/g, c/f, d/e}
            \draw[] (\u) to [bend left] (\v);
    \end{scope}
    
    \draw (.5, -.75) rectangle (2.45, -.2);
    \draw (4.5, -.75) rectangle (6.5, -.2);
\end{tikzpicture} 
\begin{tikzpicture}[scale=.75,transform shape,baseline]
\draw[->, >=stealth] (0, 0) -- (1, 0);
\end{tikzpicture} 
\begin{tikzpicture}[scale=.75,transform shape,baseline]

    \foreach [count=\i] \color in {black,black,black,black}
        \draw[color=\color,very thick] (2*\i-2, 0) -- (2*\i-1, 0);

    \foreach [count=\i from 0] \name in {a,f,g,d,e,b,c,h}
    \node[vertex] (\name) at (\i,0) [draw,circle] 
    {};
    \foreach \name/\lab in {a/$0$,b/$\pi'_1$,c/$\pi'_{2i-2}$,d/$\pi'_{2i-1}$,e/$\pi'_{2j-2}$,f/$\pi'_{2j-1}$,g/$\pi'_{2k-2}$,h/$\pi'_{2k-1}$}
    \node at (\name) [label=below:\strut\lab] {};

    \begin{scope}[color=gray!50, very thick, dotted]
        \foreach \u/\v in {a/h, g/b, f/c, d/e}
            \draw[] (\u) to [bend left] (\v);
    \end{scope}
\end{tikzpicture} 

\item 

  \begin{tikzpicture}[scale=.75,transform shape,baseline]

    \foreach [count=\i] \color in {black,black,black,black}
        \draw[color=\color,very thick] (2*\i-2, 0) -- (2*\i-1, 0);

    \foreach [count=\i from 0] \name in {a,b,c,d,e,f,g,h}
    \node[vertex] (\name) at (\i,0) [draw,circle] 
    {};
    \foreach \name/\lab in {a/$0$,b/$\pi'_1$,c/$\pi'_{2i-2}$,d/$\pi'_{2i-1}$,e/$\pi'_{2j-2}$,f/$\pi'_{2j-1}$,g/$\pi'_{2k-2}$,h/$\pi'_{2k-1}$}
    \node (\name2) at (\name) [label=below:\strut\lab] {};


    \begin{scope}[color=gray!50, very thick, dotted]
        \foreach \u/\v in {a/b, c/f, d/g, e/h}
            \draw[] (\u) to [bend left] (\v);
    \end{scope}
    
    \draw (.5, -.75) rectangle (2.45, -.2);
    \draw (4.5, -.75) rectangle (6.5, -.2);
\end{tikzpicture} 
\begin{tikzpicture}[scale=.75,transform shape,baseline]
\draw[->, >=stealth] (0, 0) -- (1, 0);
\end{tikzpicture} 
\begin{tikzpicture}[scale=.75,transform shape,baseline]

    \foreach [count=\i] \color in {black,black,black,black}
        \draw[color=\color,very thick] (2*\i-2, 0) -- (2*\i-1, 0);

    \foreach [count=\i from 0] \name in {a,f,g,d,e,b,c,h}
    \node[vertex] (\name) at (\i,0) [draw,circle] 
    {};
    \foreach \name/\lab in {a/$0$,b/$\pi'_1$,c/$\pi'_{2i-2}$,d/$\pi'_{2i-1}$,e/$\pi'_{2j-2}$,f/$\pi'_{2j-1}$,g/$\pi'_{2k-2}$,h/$\pi'_{2k-1}$}
    \node at (\name) [label=below:\strut\lab] {};

    \begin{scope}[color=gray!50, very thick, dotted]
        \foreach \u/\v in {a/b, f/c, g/d, e/h}
            \draw[] (\u) to [bend left] (\v);
    \end{scope}
\end{tikzpicture}

\item 

  \begin{tikzpicture}[scale=.75,transform shape,baseline]

    \foreach [count=\i] \color in {black,black,black,black}
        \draw[color=\color,very thick] (2*\i-2, 0) -- (2*\i-1, 0);

    \foreach [count=\i from 0] \name in {a,b,c,d,e,f,g,h}
    \node[vertex] (\name) at (\i,0) [draw,circle] 
    {};
    \foreach \name/\lab in {a/$0$,b/$\pi'_1$,c/$\pi'_{2i-2}$,d/$\pi'_{2i-1}$,e/$\pi'_{2j-2}$,f/$\pi'_{2j-1}$,g/$\pi'_{2k-2}$,h/$\pi'_{2k-1}$}
    \node (\name2) at (\name) [label=below:\strut\lab] {};


    \begin{scope}[color=gray!50, very thick, dotted]
        \foreach \u/\v in {a/b, c/h, d/e, f/g}
            \draw[] (\u) to [bend left] (\v);
    \end{scope}
    
    \draw (.5, -.75) rectangle (2.45, -.2);
    \draw (4.5, -.75) rectangle (6.5, -.2);
\end{tikzpicture} 
\begin{tikzpicture}[scale=.75,transform shape,baseline]
\draw[->, >=stealth] (0, 0) -- (1, 0);
\end{tikzpicture} 
\begin{tikzpicture}[scale=.75,transform shape,baseline]

    \foreach [count=\i] \color in {black,black,black,black}
        \draw[color=\color,very thick] (2*\i-2, 0) -- (2*\i-1, 0);

    \foreach [count=\i from 0] \name in {a,f,g,d,e,b,c,h}
    \node[vertex] (\name) at (\i,0) [draw,circle] 
    {};
    \foreach \name/\lab in {a/$0$,b/$\pi'_1$,c/$\pi'_{2i-2}$,d/$\pi'_{2i-1}$,e/$\pi'_{2j-2}$,f/$\pi'_{2j-1}$,g/$\pi'_{2k-2}$,h/$\pi'_{2k-1}$}
    \node at (\name) [label=below:\strut\lab] {};

    \begin{scope}[color=gray!50, very thick, dotted]
        \foreach \u/\v in {a/b, c/h, d/e, f/g}
            \draw[] (\u) to [bend left] (\v);
    \end{scope}
\end{tikzpicture}

\item 

  \begin{tikzpicture}[scale=.75,transform shape,baseline]

    \foreach [count=\i] \color in {black,black,black,black}
        \draw[color=\color,very thick] (2*\i-2, 0) -- (2*\i-1, 0);

    \foreach [count=\i from 0] \name in {a,b,c,d,e,f,g,h}
    \node[vertex] (\name) at (\i,0) [draw,circle] 
    {};
    \foreach \name/\lab in {a/$0$,b/$\pi'_1$,c/$\pi'_{2i-2}$,d/$\pi'_{2i-1}$,e/$\pi'_{2j-2}$,f/$\pi'_{2j-1}$,g/$\pi'_{2k-2}$,h/$\pi'_{2k-1}$}
    \node (\name2) at (\name) [label=below:\strut\lab] {};


    \begin{scope}[color=gray!50, very thick, dotted]
        \foreach \u/\v in {a/h, b/e, c/d, f/g}
            \draw[] (\u) to [bend left] (\v);
    \end{scope}
    
    \draw (.5, -.75) rectangle (2.45, -.2);
    \draw (4.5, -.75) rectangle (6.5, -.2);
\end{tikzpicture} 
\begin{tikzpicture}[scale=.75,transform shape,baseline]
\draw[->, >=stealth] (0, 0) -- (1, 0);
\end{tikzpicture} 
\begin{tikzpicture}[scale=.75,transform shape,baseline]

    \foreach [count=\i] \color in {black,black,black,black}
        \draw[color=\color,very thick] (2*\i-2, 0) -- (2*\i-1, 0);

    \foreach [count=\i from 0] \name in {a,f,g,d,e,b,c,h}
    \node[vertex] (\name) at (\i,0) [draw,circle] 
    {};
    \foreach \name/\lab in {a/$0$,b/$\pi'_1$,c/$\pi'_{2i-2}$,d/$\pi'_{2i-1}$,e/$\pi'_{2j-2}$,f/$\pi'_{2j-1}$,g/$\pi'_{2k-2}$,h/$\pi'_{2k-1}$}
    \node at (\name) [label=below:\strut\lab] {};

    \begin{scope}[color=gray!50, very thick, dotted]
        \foreach \u/\v in {a/h, e/b, d/c, f/g}
            \draw[] (\u) to [bend left] (\v);
    \end{scope}
\end{tikzpicture}

\item 

  \begin{tikzpicture}[scale=.75,transform shape,baseline]

    \foreach [count=\i] \color in {black,black,black,black}
        \draw[color=\color,very thick] (2*\i-2, 0) -- (2*\i-1, 0);

    \foreach [count=\i from 0] \name in {a,b,c,d,e,f,g,h}
    \node[vertex] (\name) at (\i,0) [draw,circle] 
    {};
    \foreach \name/\lab in {a/$0$,b/$\pi'_1$,c/$\pi'_{2i-2}$,d/$\pi'_{2i-1}$,e/$\pi'_{2j-2}$,f/$\pi'_{2j-1}$,g/$\pi'_{2k-2}$,h/$\pi'_{2k-1}$}
    \node (\name2) at (\name) [label=below:\strut\lab] {};


    \begin{scope}[color=gray!50, very thick, dotted]
        \foreach \u/\v in {a/d, b/g, c/h, e/f}
            \draw[] (\u) to [bend left] (\v);
    \end{scope}
    
    \draw (.5, -.75) rectangle (2.45, -.2);
    \draw (4.5, -.75) rectangle (6.5, -.2);
\end{tikzpicture} 
\begin{tikzpicture}[scale=.75,transform shape,baseline]
\draw[->, >=stealth] (0, 0) -- (1, 0);
\end{tikzpicture} 
\begin{tikzpicture}[scale=.75,transform shape,baseline]

    \foreach [count=\i] \color in {black,black,black,black}
        \draw[color=\color,very thick] (2*\i-2, 0) -- (2*\i-1, 0);

    \foreach [count=\i from 0] \name in {a,f,g,d,e,b,c,h}
    \node[vertex] (\name) at (\i,0) [draw,circle] 
    {};
    \foreach \name/\lab in {a/$0$,b/$\pi'_1$,c/$\pi'_{2i-2}$,d/$\pi'_{2i-1}$,e/$\pi'_{2j-2}$,f/$\pi'_{2j-1}$,g/$\pi'_{2k-2}$,h/$\pi'_{2k-1}$}
    \node at (\name) [label=below:\strut\lab] {};

    \begin{scope}[color=gray!50, very thick, dotted]
        \foreach \u/\v in {a/d, g/b, c/h, f/e}
            \draw[] (\u) to [bend left] (\v);
    \end{scope}
\end{tikzpicture} 
\end{enumerate}

  \item if $\beta$ acts on three cycles, then exactly two of the black edges on which $\beta$ acts belong to the same cycle. In all cases, we have $\Delta c(\pi, \pi\beta)\in\{-2, 0\}$; again, omitted cases are symmetric, and only proper block-interchanges are considered since block-transpositions acting on three cycles decrease the number of cycles by two~\cite[Lemma 2.1 page 227]{bafna-transpositions}: 
  
  \begin{enumerate}
  \item 
    \begin{tikzpicture}[scale=.75,transform shape,baseline]

    \foreach [count=\i] \color in {black,black,black,black}
        \draw[color=\color,very thick] (2*\i-2, 0) -- (2*\i-1, 0);

    \foreach [count=\i from 0] \name in {a,b,c,d,e,f,g,h}
    \node[vertex] (\name) at (\i,0) [draw,circle] 
    {};
    \foreach \name/\lab in {a/$0$,b/$\pi'_1$,c/$\pi'_{2i-2}$,d/$\pi'_{2i-1}$,e/$\pi'_{2j-2}$,f/$\pi'_{2j-1}$,g/$\pi'_{2k-2}$,h/$\pi'_{2k-1}$}
    \node (\name2) at (\name) [label=below:\strut\lab] {};


    \begin{scope}[color=gray!50, very thick, dotted]
        \foreach \u/\v in {a/d, b/c, e/f, g/h}
            \draw[] (\u) to [bend left] (\v);
    \end{scope}
    
    \draw (.5, -.75) rectangle (2.45, -.2);
    \draw (4.5, -.75) rectangle (6.5, -.2);
\end{tikzpicture} 
\begin{tikzpicture}[scale=.75,transform shape,baseline]
\draw[->, >=stealth] (0, 0) -- (1, 0);
\end{tikzpicture} 
\begin{tikzpicture}[scale=.75,transform shape,baseline]

    \foreach [count=\i] \color in {black,black,black,black}
        \draw[color=\color,very thick] (2*\i-2, 0) -- (2*\i-1, 0);

    \foreach [count=\i from 0] \name in {a,f,g,d,e,b,c,h}
    \node[vertex] (\name) at (\i,0) [draw,circle] 
    {};
    \foreach \name/\lab in {a/$0$,b/$\pi'_1$,c/$\pi'_{2i-2}$,d/$\pi'_{2i-1}$,e/$\pi'_{2j-2}$,f/$\pi'_{2j-1}$,g/$\pi'_{2k-2}$,h/$\pi'_{2k-1}$}
    \node at (\name) [label=below:\strut\lab] {};

    \begin{scope}[color=gray!50, very thick, dotted]
        \foreach \u/\v in {a/d, b/c, f/e, g/h}
            \draw[] (\u) to [bend left] (\v);
    \end{scope}
\end{tikzpicture} 

  \item 
    \begin{tikzpicture}[scale=.75,transform shape,baseline]

    \foreach [count=\i] \color in {black,black,black,black}
        \draw[color=\color,very thick] (2*\i-2, 0) -- (2*\i-1, 0);

    \foreach [count=\i from 0] \name in {a,b,c,d,e,f,g,h}
    \node[vertex] (\name) at (\i,0) [draw,circle] 
    {};
    \foreach \name/\lab in {a/$0$,b/$\pi'_1$,c/$\pi'_{2i-2}$,d/$\pi'_{2i-1}$,e/$\pi'_{2j-2}$,f/$\pi'_{2j-1}$,g/$\pi'_{2k-2}$,h/$\pi'_{2k-1}$}
    \node (\name2) at (\name) [label=below:\strut\lab] {};


    \begin{scope}[color=gray!50, very thick, dotted]
        \foreach \u/\v in {a/b, c/f, d/e, g/h}
            \draw[] (\u) to [bend left] (\v);
    \end{scope}
    
    \draw (.5, -.75) rectangle (2.45, -.2);
    \draw (4.5, -.75) rectangle (6.5, -.2);
\end{tikzpicture} 
\begin{tikzpicture}[scale=.75,transform shape,baseline]
\draw[->, >=stealth] (0, 0) -- (1, 0);
\end{tikzpicture} 
\begin{tikzpicture}[scale=.75,transform shape,baseline]

    \foreach [count=\i] \color in {black,black,black,black}
        \draw[color=\color,very thick] (2*\i-2, 0) -- (2*\i-1, 0);

    \foreach [count=\i from 0] \name in {a,f,g,d,e,b,c,h}
    \node[vertex] (\name) at (\i,0) [draw,circle] 
    {};
    \foreach \name/\lab in {a/$0$,b/$\pi'_1$,c/$\pi'_{2i-2}$,d/$\pi'_{2i-1}$,e/$\pi'_{2j-2}$,f/$\pi'_{2j-1}$,g/$\pi'_{2k-2}$,h/$\pi'_{2k-1}$}
    \node at (\name) [label=below:\strut\lab] {};

    \begin{scope}[color=gray!50, very thick, dotted]
        \foreach \u/\v in {a/b, f/c, d/e, g/h}
            \draw[] (\u) to [bend left] (\v);
    \end{scope}
\end{tikzpicture} 
  \item 
    \begin{tikzpicture}[scale=.75,transform shape,baseline]

    \foreach [count=\i] \color in {black,black,black,black}
        \draw[color=\color,very thick] (2*\i-2, 0) -- (2*\i-1, 0);

    \foreach [count=\i from 0] \name in {a,b,c,d,e,f,g,h}
    \node[vertex] (\name) at (\i,0) [draw,circle] 
    {};
    \foreach \name/\lab in {a/$0$,b/$\pi'_1$,c/$\pi'_{2i-2}$,d/$\pi'_{2i-1}$,e/$\pi'_{2j-2}$,f/$\pi'_{2j-1}$,g/$\pi'_{2k-2}$,h/$\pi'_{2k-1}$}
    \node (\name2) at (\name) [label=below:\strut\lab] {};


    \begin{scope}[color=gray!50, very thick, dotted]
        \foreach \u/\v in {a/b, c/d, e/h, f/g}
            \draw[] (\u) to [bend left] (\v);
    \end{scope}
    
    \draw (.5, -.75) rectangle (2.45, -.2);
    \draw (4.5, -.75) rectangle (6.5, -.2);
\end{tikzpicture} 
\begin{tikzpicture}[scale=.75,transform shape,baseline]
\draw[->, >=stealth] (0, 0) -- (1, 0);
\end{tikzpicture} 
\begin{tikzpicture}[scale=.75,transform shape,baseline]

    \foreach [count=\i] \color in {black,black,black,black}
        \draw[color=\color,very thick] (2*\i-2, 0) -- (2*\i-1, 0);

    \foreach [count=\i from 0] \name in {a,f,g,d,e,b,c,h}
    \node[vertex] (\name) at (\i,0) [draw,circle] 
    {};
    \foreach \name/\lab in {a/$0$,b/$\pi'_1$,c/$\pi'_{2i-2}$,d/$\pi'_{2i-1}$,e/$\pi'_{2j-2}$,f/$\pi'_{2j-1}$,g/$\pi'_{2k-2}$,h/$\pi'_{2k-1}$}
    \node at (\name) [label=below:\strut\lab] {};

    \begin{scope}[color=gray!50, very thick, dotted]
        \foreach \u/\v in {a/b, d/c, f/g, e/h}
            \draw[] (\u) to [bend left] (\v);
    \end{scope}
\end{tikzpicture} 

  \item 
    \begin{tikzpicture}[scale=.75,transform shape,baseline]

    \foreach [count=\i] \color in {black,black,black,black}
        \draw[color=\color,very thick] (2*\i-2, 0) -- (2*\i-1, 0);

    \foreach [count=\i from 0] \name in {a,b,c,d,e,f,g,h}
    \node[vertex] (\name) at (\i,0) [draw,circle] 
    {};
    \foreach \name/\lab in {a/$0$,b/$\pi'_1$,c/$\pi'_{2i-2}$,d/$\pi'_{2i-1}$,e/$\pi'_{2j-2}$,f/$\pi'_{2j-1}$,g/$\pi'_{2k-2}$,h/$\pi'_{2k-1}$}
    \node (\name2) at (\name) [label=below:\strut\lab] {};


    \begin{scope}[color=gray!50, very thick, dotted]
        \foreach \u/\v in {a/f, b/e, c/d, g/h}
            \draw[] (\u) to [bend left] (\v);
    \end{scope}
    
    \draw (.5, -.75) rectangle (2.45, -.2);
    \draw (4.5, -.75) rectangle (6.5, -.2);
\end{tikzpicture} 
\begin{tikzpicture}[scale=.75,transform shape,baseline]
\draw[->, >=stealth] (0, 0) -- (1, 0);
\end{tikzpicture} 
\begin{tikzpicture}[scale=.75,transform shape,baseline]

    \foreach [count=\i] \color in {black,black,black,black}
        \draw[color=\color,very thick] (2*\i-2, 0) -- (2*\i-1, 0);

    \foreach [count=\i from 0] \name in {a,f,g,d,e,b,c,h}
    \node[vertex] (\name) at (\i,0) [draw,circle] 
    {};
    \foreach \name/\lab in {a/$0$,b/$\pi'_1$,c/$\pi'_{2i-2}$,d/$\pi'_{2i-1}$,e/$\pi'_{2j-2}$,f/$\pi'_{2j-1}$,g/$\pi'_{2k-2}$,h/$\pi'_{2k-1}$}
    \node at (\name) [label=below:\strut\lab] {};

    \begin{scope}[color=gray!50, very thick, dotted]
        \foreach \u/\v in {a/f, e/b, d/c, g/h}
            \draw[] (\u) to [bend left] (\v);
    \end{scope}
\end{tikzpicture}

  \item 
    \begin{tikzpicture}[scale=.75,transform shape,baseline]

    \foreach [count=\i] \color in {black,black,black,black}
        \draw[color=\color,very thick] (2*\i-2, 0) -- (2*\i-1, 0);

    \foreach [count=\i from 0] \name in {a,b,c,d,e,f,g,h}
    \node[vertex] (\name) at (\i,0) [draw,circle] 
    {};
    \foreach \name/\lab in {a/$0$,b/$\pi'_1$,c/$\pi'_{2i-2}$,d/$\pi'_{2i-1}$,e/$\pi'_{2j-2}$,f/$\pi'_{2j-1}$,g/$\pi'_{2k-2}$,h/$\pi'_{2k-1}$}
    \node (\name2) at (\name) [label=below:\strut\lab] {};


    \begin{scope}[color=gray!50, very thick, dotted]
        \foreach \u/\v in {a/h, b/g, c/d, e/f}
            \draw[] (\u) to [bend left] (\v);
    \end{scope}
    
    \draw (.5, -.75) rectangle (2.45, -.2);
    \draw (4.5, -.75) rectangle (6.5, -.2);
\end{tikzpicture} 
\begin{tikzpicture}[scale=.75,transform shape,baseline]
\draw[->, >=stealth] (0, 0) -- (1, 0);
\end{tikzpicture} 
\begin{tikzpicture}[scale=.75,transform shape,baseline]

    \foreach [count=\i] \color in {black,black,black,black}
        \draw[color=\color,very thick] (2*\i-2, 0) -- (2*\i-1, 0);

    \foreach [count=\i from 0] \name in {a,f,g,d,e,b,c,h}
    \node[vertex] (\name) at (\i,0) [draw,circle] 
    {};
    \foreach \name/\lab in {a/$0$,b/$\pi'_1$,c/$\pi'_{2i-2}$,d/$\pi'_{2i-1}$,e/$\pi'_{2j-2}$,f/$\pi'_{2j-1}$,g/$\pi'_{2k-2}$,h/$\pi'_{2k-1}$}
    \node at (\name) [label=below:\strut\lab] {};

    \begin{scope}[color=gray!50, very thick, dotted]
        \foreach \u/\v in {a/h, g/b, d/c, f/e}
            \draw[] (\u) to [bend left] (\v);
    \end{scope}
\end{tikzpicture} 

\end{enumerate}
  
  \item if $\beta$ acts on four cycles, then all black edges on which $\beta$ acts belong to their own distinct cycle and $\Delta c(\pi, \pi\beta)=-2$:

  \begin{center}
    \begin{tikzpicture}[scale=.75,transform shape,baseline]

    \foreach [count=\i] \color in {black,black,black,black}
        \draw[color=\color,very thick] (2*\i-2, 0) -- (2*\i-1, 0);

    \foreach [count=\i from 0] \name in {a,b,c,d,e,f,g,h}
    \node[vertex] (\name) at (\i,0) [draw,circle] 
    {};
    \foreach \name/\lab in {a/$0$,b/$\pi'_1$,c/$\pi'_{2i-2}$,d/$\pi'_{2i-1}$,e/$\pi'_{2j-2}$,f/$\pi'_{2j-1}$,g/$\pi'_{2k-2}$,h/$\pi'_{2k-1}$}
    \node (\name2) at (\name) [label=below:\strut\lab] {};


    \begin{scope}[color=gray!50, very thick, dotted]
        \foreach \u/\v in {a/b, c/d, e/f, g/h}
            \draw[] (\u) to [bend left] (\v);
    \end{scope}
    
    \draw (.5, -.75) rectangle (2.45, -.2);
    \draw (4.5, -.75) rectangle (6.5, -.2);
\end{tikzpicture} 
\begin{tikzpicture}[scale=.75,transform shape,baseline]
\draw[->, >=stealth] (0, 0) -- (1, 0);
\end{tikzpicture} 
\begin{tikzpicture}[scale=.75,transform shape,baseline]

    \foreach [count=\i] \color in {black,black,black,black}
        \draw[color=\color,very thick] (2*\i-2, 0) -- (2*\i-1, 0);

    \foreach [count=\i from 0] \name in {a,f,g,d,e,b,c,h}
    \node[vertex] (\name) at (\i,0) [draw,circle] 
    {};
    \foreach \name/\lab in {a/$0$,b/$\pi'_1$,c/$\pi'_{2i-2}$,d/$\pi'_{2i-1}$,e/$\pi'_{2j-2}$,f/$\pi'_{2j-1}$,g/$\pi'_{2k-2}$,h/$\pi'_{2k-1}$}
    \node at (\name) [label=below:\strut\lab] {};

    \begin{scope}[color=gray!50, very thick, dotted]
        \foreach \u/\v in {a/b, d/c, f/e, g/h}
            \draw[] (\u) to [bend left] (\v);
    \end{scope}
\end{tikzpicture} 

  \end{center}
 \end{enumerate}
\end{proof}

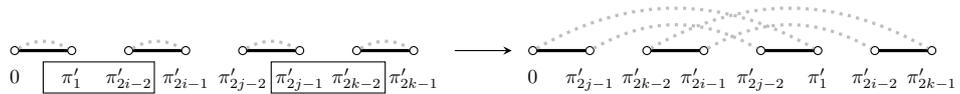
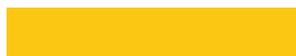
\end{document}